\documentclass{article} 
\let\savedbbchi\bbchi
\usepackage{graphicx}
\usepackage{color}
\usepackage{amsmath,amsthm}
\usepackage{amsfonts}
\usepackage[normal]{subfigure}
\usepackage{a4wide}
\usepackage{enumitem}
\usepackage{natbib}
\usepackage{authblk}

\renewcommand\harvardyearright[1]{.} 
\let\bbchi\savedbbchi

\newcommand{\Ub}{{\bf U}}
\newcommand{\vb}{{\bf v}}
\newcommand{\hv}{\hat{\bf v}}

\newcommand{\x}{{\bf x}}

\newcommand{\ie}{{\it i.e.}}
\newcommand{\eg}{{\it e.g.}}

\newtheorem{theorem}{Theorem}[section]
\newtheorem{proposition}[theorem]{Proposition}

\begin{document}

\title{Stability of a non-local kinetic model for cell migration with density dependent orientation bias}


\author{Nadia Loy \thanks{Department of Mathematical Sciences ``G. L. Lagrange'', Politecnico di Torino, Corso Duca degli Abruzzi 24, 10129 Torino, Italy, and Department of Mathematics ``G. Peano'', Via Carlo Alberto 10 ,10123 Torino, Italy
                (\texttt{nadia.loy@polito.it})}\and
        Luigi Preziosi\thanks{Department of Mathematical Sciences ``G. L. Lagrange'', Dipartimento di Eccellenza 2018-2022, Politecnico di Torino, Corso
                Duca degli Abruzzi 24, 10129 Torino, Italy
                (\texttt{luigi.preziosi@polito.it})}
                \thanks{Corresponding author: \texttt{nadia.loy@polito.it}}}
                
\maketitle

\begin{abstract}
The aim of the article is to study the stability of a non-local kinetic model 
proposed by \cite{Loy_Preziosi}. We split the population in two subgroups and perform a linear stability analysis. We show that pattern formation results from modulation of one non-dimensional parameter that depends on the tumbling frequency, the sensing radius, the mean speed in a given direction, the uniform configuration density and the tactic response to the cell density. Numerical simulations show that our linear stability analysis predicts quite precisely the ranges of parameters determining instability and pattern formation. We also extend the stability analysis in the case of different mean speeds in different directions. In this case, for parameter values leading to instability travelling wave patterns develop.

\end{abstract}

\section{Introduction}

Cells decide where to go by sensing the surrounding environment, extending their protrusions over a distance that can reach several cell diameters. 
In order to describe such a non-local action several mathematical models have been recently proposed in the literature.

\cite{Othmer_Hillen.02} and \cite{Hillen_Painter_Schmeiser.06} introduced a finite sampling radius and defined a non-local gradient as the average of the external field on a surface which represents the membrane of the cell.  In particular, the Authors derived a macroscopic diffusive model with a non-local gradient both from a position jump process and from a velocity jump process by postulating that the non-local sensing is a bias of higher order.

With the aim of modelling cell-cell adhesion and haptotaxis \cite{Armstrong_Painter_Sherratt.06} proposed a macroscopic integro-differential equation where the integral over a finite radius is in charge of describing non-local sensing. Recently, \cite{Butt} derived this model from a space jump process, while \cite{Buttenschoen_Thesis}, \cite{Butt_20} studied the steady states and bifurcations of the macroscopic adhesion model and their stability. Further studies concerning these models on bounded domains are proposed in \citep{Butt_19}. 
Other macroscopic models describing cell migration with non-local measures of the environment were proposed by \cite{Painter_Sherratt.2010, Painter.Gerish.15} and \cite{Painter_Hillen.02}. \cite{Schmeiser_Nouri} considered a kinetic model with velocity jumps biased towards the chemical concentration gradient.  Similar equations were also proposed in 2D set-ups by \cite{Col_Sci_Tos.15, Col_Sci_Prez.17},  and applied to model crowd dynamics and traffic flow for instance by \cite{Tosin_Frasca}. 

\cite{Eftimie_deVries_Lewis} proposed a non-local kinetic models including repulsion, alignement and attraction. They distinguish cells going along the two directions of a one-dimensional set-up with the possibility of switching between the two directions keeping the same speed (they thus have two velocities). The one-dimensional kinetic model takes then a discrete velocity structure. Its integration shows a wide variety of patterns. The linear stability analysis of the model is then performed in \cite{Eftimie.deVries.Lewis}. \cite{Eftimie} derived the macroscopic limits of this model. In \cite{Eftimie.17} and  \cite{Bitsouni2018} the model was applied to model,  respectively, tumour dynamics and cell polarisation in heterogeneous cancer cell populations. The linear stability analysis of a kinetic chemotaxis equation coupled with a macroscopic diffusion equation for the chemical dynamics is presented by \cite{Perthame_2018}. 

\cite{Loy_Preziosi} proposed a non-local kinetic model with double-bias on the basis of the observation that  different fields can influence respectively cell polarization and speed. Then the turning operator is characterized by the presence of two integrals, one determining the probability of polarizing in a certain direction after averaging the sensing of a chemical or mechanical cue over a finite neighborhood, and the other setting the probability of moving with a certain speed in the chosen direction after averaging the sensing of another chemical or mechanical cue over a possibly different finite radius.
The model was then modified in \cite{Loy_Preziosi2} to take into account of physical limits of migrations, that might hamper the real possibility of cells of measuring beyond physical barriers or to move in certain regions because of physically constraining situations, such as too dense extracellular matrix or cell overcowding.
In such papers it was shown that the presence of density dependent cues, such as cell-cell adhesion and volume filling, may generate instabilities and the formation of patterns.
In fact, on the one hand cells may be attracted due to the mutual interaction of transmembrane adhesion molecules ($\eg$, cadherin complexes). On the other hand, they may want to stay away from overcrowded areas. 
It seems that, depending on the sensing kernel and on other modelling parameters, both the wish to stay together and to stay away might lead to instability. 

In this paper we study the stability of the homogeneous configuration when density dependent cues influence cell polarization. The case in which the density distribution also affects cell speed will be treated in a following paper. 

With this aim in mind, in Section 2 we briefly recall the non-local kinetic model, then focalizing to the case of orientational biases.
Restricting to the one-dimensional case, in Section 3 the linear stability analysis is performed, first for a general speed distribution function and then in the particular case of a Dirac delta. At this stage the sensing kernel is still general and it is proved that if it is a non-increasing function of the distance, then staying away strategies are always stable. On the contrary,  Section 3.1 shows that a localized sensing kernel, $\ie$ a Dirac delta, might lead to instability of a finite wavelength if for instance, cell speed is sufficiently small, or the turning rate or the sensing radius is sufficiently high, or in the case of volume filling effects, if cell density is sufficiently high. Section 3.2 and Section 3.3 then respectively focus on a homogeneous and a decreasing sensing kernel over a finite range, specifically a Heaviside function and a ramp function. It these cases, as well as for the localized sensing kernel, cell-cell adhesion strategies lead to long wave instabilities. Section 4 reports some simulations of the cases above, while Section 5 discusses the case of an asymmetric speed distribution function in the two directions. In the unstable case this leads to the formation of travelling waves instabilities similar to those reported by \cite{Eftimie.deVries.Lewis}, \cite{Eftimie_deVries_Lewis}, \cite{Eftimie}, \cite{Eftimie2}.
A final section  draws some conclusions pointing out  possible developments.

\section{The model}

In the model introduced by \cite{Loy_Preziosi}
the cell population is described at a mesoscopic level by the distribution density $p = p(t,\x, v,\hv)$ parametrized by the time $t>0$, the position $\x \in \Omega \subseteq \mathbb{R}^d$, the speed $v\in \mathbb{R}_+$ and the polarization direction $\hv\in \mathbb{S}^{d-1}$ where  $\mathbb{S}^{d-1}$ is the unit sphere boundary in $\mathbb{R}^d$.  
We remark that the distribution function $p$ depends separately on velocity modulus and direction, instead of the velocity vector $\vb=v\hv$. This is due to the need  of separating the subcellular  mechanisms governing cell polarization and motility. In fact, cells respond both to tactic factors affecting the choice of the direction, and to kinetic factors, typically of mechanical origins, influencing cell speed.

The mesoscopic model consists in the transport equation for the cell distribution
\begin{equation}\label{transport.general}
\dfrac{\partial p}{\partial t}(t,\x,v,\hv) + \vb\cdot \nabla p(t,\x,v,\hv) =  \mathcal{J} [p] (t,\x,v,\hv)
\end{equation}
where the operator $\nabla$ denotes the spatial gradient then
coupled with proper initial  and boundary conditions. In particular, we shall consider no-flux boundary conditions that are defined as \citep{Plaza}
\begin{equation}\label{noflux}
\int_{\mathbb{R}_+} \int_{\mathbb{S}^{d-1}} p(t,\x,v,\hv)\vb\cdot {\bf n}(\x)d\hv \, dv=0, \quad \forall \x \in \partial \Omega, \quad t>0
\end{equation}
being ${\bf n}(\x)$ the outer normal to the boundary $\partial \Omega$ in the point $\x$. 
Equation \eqref{noflux} implies that there is no mass flux across the boundary \citep{Lemou}.

A macroscopic description for the cell population can be classically recovered  through the definition of moments of the distribution function $p$. For instance,  the cell number density will be given by
\begin{equation}\label{def_rho}
\rho(t,\x) = \displaystyle\int_{\mathbb{S}^{d-1}}\int_{\mathbb{R}_+} p(t,\x,v,\hv) \,dv\,d\hv\,,
\end{equation}
and the cell mean velocity by
\begin{equation}\label{mean.U}
\Ub(t,\x) = \dfrac{1}{\rho(t,\x)}\displaystyle\int_{\mathbb{S}^{d-1}}\int_{\mathbb{R}_+} p(t,\x,v,\hv)\vb \,dv\,d\hv\,.
\end{equation}



The term $\mathcal{J}[p](t,\x,v,\hv)$, named {\it turning operator},
is an  integral operator that describes the change in velocity which is not due to free-particle transport. It may describe the classical run and tumble behaviors,  random re-orientations, which, however, may be biased  by external cues. In the present case, the turning operator will be the implementation of a velocity-jump process in a kinetic transport equation as introduced by \cite{Stroock} and then by \cite{Alt.88}. 
The turning operator that  we are going to consider is
\begin{equation}\label{J_r.S}
 \mathcal{J}[p](t,\x,v,\hv) = \mu (\x) \, \Big( \rho(t,\x)  T(\x,v,\hv) - p(t,\x,v,\hv) \Big) \,,
\end{equation}
which is obtained assuming that cells retain no memory of their velocity prior to the re-orientation, and also turning rates do not depend on the orientation of the individual cell.

Following \cite{Loy_Preziosi}, we consider here a transition probability that depends on the non-local sensing of the macroscopic density of the cell population in a neighborhood of the cell, that writes as
\begin{equation}\label{distribution.g.r}
T[\rho](\x,v,\hv)=c(t,\x)\int_{\mathbb{R}_+}\gamma_R(\lambda) b\big(\rho(t,\x + \lambda\hv)\big)\, d\lambda \,\psi(v|\hv),
\end{equation}
where $c(t,\x)$ is a normalization constant, that is 
\[
c^{-1}(t,\x)=\displaystyle\int_{\mathbb{S}^{d-1}}\left[\int_{\mathbb{R}_+}b\left(\rho(t,\x+\lambda\hv)\right)\gamma_R(\lambda)d\lambda\right]\psi(v|\hv) \, d\hv\,,
\]
so that the integral of $T$ over the velocity space is one. In this way, $T$ is a mass preserving transition probability.

The term $b$ describes the response of the cells to the tactic cue, in this case the cell density $\rho$ itself, around $\x$ along the direction $\hv$ and, therefore, the bias intensity in the direction $\hv$  while the sensing kernel $\gamma_R(\lambda)$ weights the collected density signal with respect to the distance $\lambda$ from $\x$. 
In particular,  $\gamma_{R}$, that we shall assume to be a positive valued $L^1(\mathbb{R}_+)$ function, has a compact support in $[0,R]$ where $R$ is the maximum extension of cell protrusions, determining the furthest points cells can reach to measure the external signals. So, integrals over $\mathbb{R}_+$ are actually integrals over the finite interval $[0,R]$.  Specifically, if 
\begin{itemize}
\item  $\gamma_R(\lambda)=\delta(\lambda-R)$ is a Dirac delta, then cells only measure the information perceived on a spherical surface of given radius $R$,
\item if $\gamma_R(\lambda)=H (R-\lambda)$ is a Heaviside function, then cells explore the whole volume of the sphere centered in $\x$ with radius $R$ and weight the information uniformly, or 
\item if $\gamma_R(\lambda)$ is a decreasing function of $\lambda$, then closer information play a bigger role with respect to farther ones, taking for instance into account that the probability of making longer protrusions decreases with the distance, so the sensing of closer regions is more accurate. 
\end{itemize}

In conclusion, the integral in \eqref{distribution.g.r} is such that if the signal is stronger in the direction $\hat\vb$, then there will be a higher probability for the cell to move along $\hat\vb$ than along $-\hat\vb$. Other functional forms for the $b$ term could be considered as discussed by \cite{Loy_Preziosi}.

The function $\psi=\psi(v|\hv)$ is the density distribution of the speeds that we assume to depend on the direction $\hv \in \mathbb{S}^{d-1}$. In fact,  \cite{Loy_Preziosi} introduce a density distribution $\psi$ describing the probability of having a certain speed in a given direction given the non local sensing of a tactic external cue in that direction. Therefore, $\psi$ also depends on the spatial variable through the kinetic cue. For simplicity we will drop this extra dependency.

In the following the mean of the probability density function $\psi$ will be denoted by $V$ and its variance by $s^2$. As $\psi$ depends on the direction $\hv$, $V$ and $s^2$ will depend on the direction $\hv$ as well.

\section{Linear stability analysis}\label{Sec:lin.inst}

In order to perform a stability analysis of the uniform configuration, we will assume that $\mu$ is constant and for the moment we will start assuming that $\psi=\psi(v|\hv)$ does not depend on $\hv$.  We will treat the case in which $\psi=\psi(v|\hv)$ in Section 5. Hence, we consider the equation 
\begin{equation}\label{eq:dir.1D}
\dfrac{\partial p}{\partial t}(t,\x,v,\hv) + v\hv\cdot \nabla p(t,\x,v,\hv) =  \mu\left[ 
c(t,\x)\rho(t,\x)\displaystyle\int_{\mathbb{R}_+}\gamma_R(\lambda)b\left(\rho(t,\x+\lambda\hv)\right)\,d\lambda \,\psi(v)
-p(t,\x,v,\hv)\right].
\end{equation}
We will also carry out the analysis in the one-dimensional case and call ${\bf e}$ and  $-{\bf e}$ the two possible directions characterizing the one dimensional problem. The population is then splitted into two subgroups $p^+$ and $p^-$ corresponding to the groups of cells respectively going to the right and to the left, so that
\begin{equation}\label{def:p}
p(t,x,v,\hv)= p^+(t,x,v)\delta(\hv-{\bf e})+p^-(t,x,v)\delta(\hv+{\bf e})\,.
\end{equation}
Coherently, we define 
\[
\rho(t,x)=\rho^+(t,x)+\rho^-(t,x)
\quad{\rm with}\quad \rho^\pm(t,x)=\int_{\mathbb{R}_+} p^\pm(t,x,v)\,dv\,.
\]
The system of equations satisfied by $p^+$ and $p^-$ is
\begin{equation}\label{eq:sist.p_gen}
\begin{array}{lc}
\dfrac{\partial p^+}{\partial t}(t,x,v) + v\dfrac{\partial p^+}{\partial x}(t,x,v) =  \mu\left[ \rho(t,x)
T^+[\rho](v)-p^+(t,x,v)\right]\\[12pt]
\dfrac{\partial p^-}{\partial t}(t,x,v) - v\dfrac{\partial p^-}{\partial x}(t,x,v) =  \mu\left[ \rho(t,x)
T^-[\rho](v)-p^-(t,x,v)\right],
\end{array}
\end{equation}
where
\begin{equation}\label{T+}
T^+[\rho](v)=\dfrac{\displaystyle\int_{\mathbb{R}_+}\gamma_R(\lambda)b\left(\rho(t,x+\lambda)\right)\,d\lambda}
{\displaystyle\int_{\mathbb{R}_+}\gamma_R(\lambda)\left[b(\rho(t,x+\lambda))+b(\rho(t,x-\lambda))\right]d\lambda}\psi(v)
\end{equation}
and
\begin{equation}\label{T-}
T^-[\rho](v)=\dfrac{\displaystyle\int_{\mathbb{R}_+}\gamma_R(\lambda)b\left(\rho(t,x-\lambda)\right)\,d\lambda}
{\displaystyle\int_{\mathbb{R}_+}\gamma_R(\lambda)\left[b(\rho(t,x+\lambda))+b(\rho(t,x-\lambda))\right]d\lambda}.
\psi(v)
\end{equation}
It can be proved that the local asymptotic equilibrium states of \eqref{eq:sist.p_gen} are \citep{Loy_Preziosi}
\begin{equation}\label{asympt}
p^+_{\infty}=\rho_{\infty}T^+[\rho_{\infty}](v), \qquad p^-_{\infty}=\rho_{\infty}T^-[\rho_{\infty}](v).
\end{equation}
They are  homogeneous if and only if 
\begin{equation}\label{homo}
p^+_{\infty}(v)=p^-_{\infty}(v)=\rho_{\infty}\dfrac{\psi(v)}{2}.
\end{equation}
Hence, once $\psi$ is chosen, all the possible spatially homogeneous and stationary solutions are determined by the choice of $\rho_{\infty}$.
We explicitly notice that the macroscopic densities will be $\rho^+_{\infty}=\rho^-_{\infty}=\dfrac{\rho_{\infty}}{2}$. 

In order to  perform the stability analysis of this configuration, we consider a small perturbation of the homogeneous solution \eqref{homo} as
\[
p^+(t,x,v)=p^+_{\infty}(v)+\hat{p}^{+}(t,x,v), \qquad p^-(t,x,v)=p^-_{\infty}(v)+\hat{p}^{-}(t,x,v).
\]
Hence,
\[
\rho(t,x)=\rho_{\infty}+ \displaystyle\int_{\mathbb{R}_+} \int_{\mathbb{S}^{d-1}} \left[ \hat{p}^+(t,x,v)\delta(\hv-{\bf e})+\hat{p}^-(t,x,v)\delta(\hv+{\bf e})\right]d\hv \,dv = \rho_{\infty}+\hat{\rho}^+(t,x)+\hat{\rho}^-(t,x)
\]
where $\hat{\rho}^{\pm}(t,x)=\displaystyle\int_{\mathbb{R}_+} \hat{p}^{\pm}(t,x,v) dv$ and  we define $\hat{\rho}=\hat{\rho}^++\hat{\rho}^-$.
Thus, the system \eqref{eq:sist.p_gen}  becomes
\[
\begin{split}\label{eq:sis.p.lin}
\dfrac{\partial \hat{p}^+}{\partial t}(t,x,v) + v\dfrac{\partial \hat{p}^+}{\partial x} (t,x,v) =  
\mu\left[\big(\rho_{\infty}+\hat{\rho}\big)T^+[\rho_{\infty}+\hat{\rho}](v)-p^+_{\infty}(v)-\hat{p}^+(t,x,v)\right],\\
\dfrac{\partial \hat{p}^-}{\partial t}(t,\x,v) - v\dfrac{\partial \hat{p}^+}{\partial x} (t,x,v) =  
\mu\left[\big(\rho_{\infty}+\hat{\rho}\big)T^-[\rho_{\infty}+\hat{\rho}](v)-p^+_{\infty}(v)-\hat{p}^-(t,x,v)\right].
\end{split}
\]
We now need to linearize the transition probabilities $T^{\pm}[\rho_{\infty}+\hat{\rho}](v)$ where, for instance,
\[
T^+[\rho_{\infty}+\hat{\rho}](v)=\dfrac{\displaystyle\int_{\mathbb{R}_+}b(\rho_{\infty}+\hat{\rho}(t,x+\lambda))\gamma_R(\lambda) \, d\lambda}{\displaystyle\int_{\mathbb{R}_+} \left[b(\rho_{\infty}+\hat{\rho}(t,x+\lambda))+b(\rho_{\infty}+\hat{\rho}(t,x-\lambda))\right]\gamma_R(\lambda) \, d\lambda}\psi(v)\,.
\]
Assuming $\hat{\rho}$ small (and then neglecting perturbation of higher order), we may perform a Taylor expansion and write
\begin{equation}
\begin{array}{cl}
T^+[\rho_{\infty}+\hat{\rho}](v)&\approx
\dfrac{\displaystyle\int_{\mathbb{R}_+}\left[b(\rho_{\infty})+b'(\rho_{\infty})\hat{\rho}(x+\lambda)\right]\gamma_R(\lambda) \, d\lambda}{\displaystyle\int_{\mathbb{R}_+}\left[b(\rho_{\infty})+b'(\rho_{\infty})\hat{\rho}(x+\lambda)+b(\rho^{\infty})+b'(\rho_{\infty})\hat{\rho}(x-\lambda)\right]\gamma_R(\lambda) \, d\lambda}\psi(v)
\\[25pt]
&\approx \dfrac{1+\dfrac{b'(\rho_{\infty})}{b(\rho_{\infty})\Gamma_R}\displaystyle \int_{\mathbb{R}_+}\hat{\rho}(t,x+\lambda)\gamma_R(\lambda) \, d\lambda}{1+\dfrac{b'(\rho_{\infty})}{b(\rho_{\infty})\Gamma_R}\displaystyle\int_{\mathbb{R}_+}\dfrac{\hat{\rho}(t,x+\lambda)+\hat{\rho}(t,x-\lambda)}{2}\gamma_R(\lambda) \, d\lambda}\dfrac{\psi(v)}{2}
\end{array}
\end{equation}
being $\Gamma_R=\displaystyle \int_{\mathbb{R}_+} \gamma_R(\lambda) d\lambda$. Expanding the denominator, we eventually have

\begin{equation}
T^+[\rho_{\infty}+\hat{\rho}](v)\approx
\left[ 1+\dfrac{b'(\rho_{\infty})}{b(\rho_{\infty})\Gamma_R}\displaystyle\int_{\mathbb{R}_+}\dfrac{\left[\hat{\rho}(t,x+\lambda)-\hat{\rho}(t,x-\lambda)\right] }{2}\gamma_R(\lambda)\, d\lambda\right]\dfrac{\psi(v)}{2}\,.
\end{equation}

Analogously
\begin{equation}
T^-[\rho_{\infty}+\hat{\rho}](v) \approx  \left[ 1-\dfrac{b'(\rho_{\infty})}{b(\rho_{\infty})\Gamma_R}\displaystyle\int_{\mathbb{R}_+}\dfrac{\left[\hat{\rho}(t,x+\lambda)-\hat{\rho}(t,x-\lambda)\right] }{2}\gamma_R(\lambda)\, d\lambda\right]\dfrac{\psi(v)}{2}\,.
\end{equation}
Therefore, the right hand sides of the system \eqref{eq:sis.p.lin} become
\[
\begin{array}{cl}
&\big(\rho_{\infty}+\hat{\rho}(t,x)\big)T^{\pm}[\rho_{\infty}+\hat{\rho}](v)-p^{\pm}_{\infty}(v)-\hat{p}^{\pm}(t,x,v)\\
&\approx \overbrace{\rho_{\infty}\dfrac{\psi(v)}{2}-p^{\pm}_{\infty}(v)}^{=0}+
\left[ \hat{\rho}(t,x)\pm \rho_{\infty}\dfrac{b'(\rho_{\infty})}{b(\rho_{\infty})\Gamma_R}\displaystyle\int_{\mathbb{R}_+}\dfrac{\left[\hat{\rho}(t,x+\lambda)-\hat{\rho}(t,x-\lambda)\right] }{2}\gamma_R(\lambda) d\lambda\right]\dfrac{\psi(v)}{2}-\hat{p}^{\pm}(v)\,.
\end{array}
\]
Let us now consider perturbations in the form
\[
\hat{p}^{\pm}(t,x,v)=g^{\pm}(v)e^{ikx+\sigma t}
\]
where $g^{\pm}$ have densities defined as $\rho_{g^{\pm}}=\int_{\mathbb{R}_+} g^{\pm}(v)dv$ and, then, $\hat{\rho}^{\pm}=\rho_{g^{\pm}}e^{ikx+\sigma t}$. 

Substitution in \eqref{eq:sis.p.lin} leads to
\begin{equation}\label{eq:lin:int}
\sigma g^{\pm} \pm i kv g^{\pm}+\mu g^{\pm}
=\mu\left[ \rho_g\pm \rho_{\infty}\dfrac{b'(\rho_{\infty})}{b(\rho_{\infty})\Gamma_R}\rho_g
\displaystyle \int_{\mathbb{R}_+}\dfrac{e^{ik\lambda}-e^{-ik\lambda}}{2} \gamma_R(\lambda)d\lambda\right]\dfrac{\psi(v)}{2}\,,
\end{equation}
where we set $\rho_g=\rho_g^++\rho_g^-$.
Now, as $e^{ik\lambda}-e^{-ik\lambda}=2i\sin(k\lambda)$,
defining 
\begin{equation}\label{gamma_trans}
\hat{\gamma}_R(k)=\dfrac{1}{\Gamma_R}\int_{\mathbb{R}_+}\sin(k\lambda)\gamma_R(\lambda)d\lambda\,,
\end{equation}
the normalized unilateral sine transform of $\gamma_R$, the system \eqref{eq:lin:int} rewrites as
\begin{equation}\label{eq:sis.g}
\begin{cases}
\left( \sigma +ikv+\mu\right)g^+=\mu\left[1+
i\mathcal{B} \hat{\gamma}_R(k) \right]\dfrac{\psi(v)}{2}\displaystyle\int_{\mathbb{R}_+}\big[ g^+(v)+g^-(v)\big]dv,\\[12pt]
\left( \sigma -ikv+\mu\right)g^-=\mu\left[1-
i\mathcal{B} \hat{\gamma}_R(k) \right]\dfrac{\psi(v)}{2}\displaystyle\int_{\mathbb{R}_+}\big[ g^+(v)+g^-(v)\big]dv,
\end{cases}
\end{equation}
where 
\begin{equation}\label{B}
\mathcal{B}=\dfrac{\rho_{\infty}b'(\rho_{\infty})}{b(\rho_{\infty})}\,.
\end{equation}
We remark that the dimensionless number  $\mathcal{B}$ can be either negative or positive, according to the fact that $b$ is a decreasing or an increasing function of $\rho$. 
From the phenomenological point of view the former case ($\mathcal{B}<0$ or $b'(\rho_\infty)<0$)
corresponds to a predisposition of cells at a density $\rho_{\infty}$ to re-polarize toward regions that have lower cell densities and move away from crowded areas, the latter case ($\mathcal{B}>0$ or $b'(\rho_\infty)>0$) 
corresponds to a predisposition of cells to re-orient toward regions with a density higher than $\rho_{\infty}$, $\eg$ a sort of adhesion-like behaviour due to the fact that cells wants to stay together.  

We observe that by summing the two equations in \eqref{eq:sis.g} we readily have
\[
\big(\sigma +\mu \big)\big(g^++g^-\big)+ikv\big(g^+-g^-\big)=\mu \psi(v) \int_{\mathbb{R}_+}\left[ g^+(v)+g^-(v)\right]dv.
\]
If we then integrate over ${\mathbb{R}_+}$, by defining $\rho_g U_g=\displaystyle \int_{\mathbb{R}_+}\left(g^+-g^-\right)v\,dv$ the mean momentum of the perturbation, we get
\[
\left( \sigma + \mu \right)\rho_g+ik\rho_g U_g=\mu \rho_g,
\]
and, then, the mean speed of the perturbation is
\[
U_g=i\dfrac{\sigma}{k}.
\]
Looking for unstable situations (so that it is sure that we are not dividing by zero) the system \eqref{eq:sis.g} can be written as
\begin{equation}\label{eq:sis.g.2}
\begin{cases}
g^+=\mu\dfrac{1+i\mathcal{B} \hat{\gamma}_R(k)} {\sigma +ikv+\mu}\dfrac{\psi(v)}{2}
\displaystyle\int_{\mathbb{R}_+}\big[ g^+(v)+g^-(v)\big]dv,\\[12pt]
g^-=\mu\dfrac{1-i\mathcal{B} \hat{\gamma}_R(k)} {\sigma -ikv+\mu}\dfrac{\psi(v)}{2}
\displaystyle\int_{\mathbb{R}_+}\big[ g^+(v)+g^-(v)\big]dv.
\end{cases}
\end{equation}

If we integrate the two equations over ${\mathbb{R}_+}$ and sum them, we obtain the following solvability condition (for non trivial solutions)
\begin{equation}
\dfrac{\mu}{2}\int_{\mathbb{R}_+}\left[
\dfrac{1+i\mathcal{B} \hat{\gamma}_R(k)} {\sigma +ikv+\mu}+
\dfrac{1-i\mathcal{B} \hat{\gamma}_R(k)} {\sigma -ikv+\mu}\right]\psi(v)\,dv=1,
\end{equation}
or 
\begin{equation}\label{solvability}
\mu\int_{\mathbb{R}_+}\dfrac{\sigma +\mu+\mathcal{B} \hat{\gamma}_R(k) kv}{(\sigma +\mu)^2+k^2v^2}\psi(v)\,dv=1.
\end{equation}

In order to give an analytical discussion of the result, let us take, 
as an example, $\psi(v)=\delta(v-V)$. In this case the integral condition \eqref{solvability} takes the algebraic form
\begin{equation}
\mu\dfrac{\sigma +\mu+\mathcal{B} \hat{\gamma}_R(k) kV}{(\sigma +\mu)^2+k^2V^2}=1\,,
\end{equation}
so that the dispersion relation reads
\[
\sigma^2+\mu\sigma +k^2V^2-\mu \mathcal{B} kV\hat{\gamma}_R(k)=0.
\]
The most dangerous eigevalue is then
\begin{equation}\label{rel.disp.2}
\sigma=\dfrac{-\mu+\sqrt{\Delta}}{2},\qquad {\rm with} \quad \Delta =
\mu^2-4k^2V^2+4\mu \mathcal{B} kV\hat{\gamma}_R(k).
\end{equation}
Therefore, the instability condition $\Re e(\sigma)>0$, given by $\Delta>0$, is 
\begin{equation}\label{inst.cond_pre}
\mathcal{B} \dfrac{\hat{\gamma}_R(k)}{k}>\dfrac{V}{\mu}.
\end{equation}
that, introducing the dimensionless number
\begin{equation}
\mathcal{V}=\dfrac{V}{\mu R}\,,
\end{equation}
 may be rewritten as
\begin{equation}\label{inst.cond}
\mathcal{B} \dfrac{\hat{\gamma}_R(k)}{Rk}>\mathcal{V}.
\end{equation}

Therefore, one can conclude that the following proposition holds.

\begin{proposition} 
If $\mathcal{B}\le 0$ and $\gamma_R(\lambda)>0$ is non increasing, then the homogeneous configuration is always stable.

 If $\mathcal{B}>0$ and $\gamma_R(\lambda)>0$ is such that $\lambda \gamma_R(\lambda)\in L^1(\mathbb{R}_+)$, then long waves $k\approx 0$ are unstable for sufficiently small values of $\mathcal{V}=\frac{V}{\mu\mathcal{B}}$.
\end{proposition}

\begin{proof}
In order to prove the statement it is enough to observe that if the sensing kernel $\gamma_R(\lambda)$ is non increasing 
then $\hat{\gamma}_R(k)> 0$  $\forall k>0$, vanishing only in the trivial cases $k=0$ or in the limit of constant sensing kernel in $\mathbb{R}_+$. 
Then condition \eqref{inst.cond} is never satisfied, being its r.h.s. strictly positive. 

On the other hand, under the stated integrability conditions
$$\lim_{k \to 0}\dfrac{\hat{\gamma}_R(k)}{k} =\dfrac{\displaystyle\int_0^\infty \lambda\gamma_R(\lambda)\,d\lambda}{\displaystyle \int_0^\infty \gamma_R(\lambda)\,d\lambda}>0.$$
Hence, if $\mathcal{B}>0$, then condition \eqref{inst.cond} is satisfied for sufficiently small ratios $\frac{V}{\mu\mathcal{B}}$ leading to instability.
\medskip
\end{proof}

We observe that the statement also holds for the relevant case of a sensing kernel with compact support, for instance 
for the Heaviside and ramp kernels that will be respectively considered in Sections \ref{sec:Heaviside} and \ref{sec:ramp}.
On the other hand, it does not hold for the localized sensing kernel $\gamma_R(\lambda)=\delta(\lambda-R)$ that will be considered in the following section. In fact, for such kernels instability is possible also for $\mathcal{B}<0$.

The second part of the proposition assures that, upon suitable integrability conditions that are satisfied by all the kernels mentioned above, there is always a value of $\frac{V}{\mu\mathcal{B}}$ that is unstable at least to long waves. 
 
We conclude this general part of the analysis by observing that in unstable situations, the maximum growth rate $\sigma_{max}=\sigma(k_{max})$ can be identified by evaluating the stationary points of 
\eqref{rel.disp.2}, that are obtained when
\begin{equation}\label{maxgrowth}
\dfrac{1}{\Gamma_R k_{max}R}
\int_{\mathbb{R}_+}[\sin (k_{max}\lambda)+k_{max}\lambda \cos (k_{max}\lambda)]\gamma_R(\lambda)\,d\lambda
=2\mathcal{V}_b,
\end{equation}
where 
\begin{equation}
\mathcal{V}_b=\mathcal{V}/\mathcal{B}.
\end{equation}

\subsection{Localized sensing kernel}
We assume now that also the sensing function is a Dirac delta, $\ie$ $\gamma_R=\delta(\lambda-R)$, so that the two equations satisfied by $p^+$ and $p^-$  \eqref{eq:sist.p_gen} specialize as
\begin{equation}\label{eq:sist.p}
\begin{split}
\dfrac{\partial p^+}{\partial t}(t,x,v) + v\dfrac{\partial p^+}{\partial x}(t,x,v) =  \mu\left[ \rho(t,x)
\dfrac{b(\rho(t,x+R))}{b(\rho(t,x+R))+b(\rho(t,x-R))}\psi(v)-p^+(t,x,v)\right],\\[12pt]
\dfrac{\partial p^-}{\partial t}(t,x,v) - v\dfrac{\partial p^-}{\partial x}(t,x,v) =  \mu\left[ \rho(t,x)
\dfrac{b(\rho(t,x-R))}{b(\rho(t,x+R)+b(\rho(t,x-R))}\psi(v)-p^-(t,x,v)\right].
\end{split}
\end{equation}
In this case, $\Gamma_R=1$ and $\hat{\gamma}_R(k)=\sin(kR)$, so that the criterium \eqref{inst.cond}
now reads
\begin{equation}\label{inst.cond.dirac}
\mathcal{B}\dfrac{\sin(kR)}{kR}>\mathcal{V}.
\end{equation}
From \eqref{maxgrowth} the maximum growth rate is obtained for $k_{max}$ such that
\begin{equation}\label{kmax_dirac}
\dfrac{\sin(k_{max}R)}{k_{max}R}+\cos(k_{max}R)=2\mathcal{V}_b.
\end{equation}

To discuss the stability properties it is now useful to distinguish two cases according to the sign of $b'$, which means the sign of $\mathcal{B}$ or of $\mathcal{V}_b$.

\subsubsection{Case $b'(\rho_{\infty})<0$}
We recall that, from the phenomenological point of view, a decreasing $b$ corresponds to a predisposition of cells to re-orient toward regions that are less crowded than the uniform stationary solution $\rho_{\infty}$. 

In this case, as $\mathcal{B}<0$, the instability condition \eqref{inst.cond.dirac} becomes
\begin{equation}\label{k_dirac_b_neg}
\dfrac{\sin(kR)}{kR}<\mathcal{V}_b<0,
\end{equation}
and the critical condition, $\ie$ the first value for which 
\begin{equation}\label{kcr_dirac}
\dfrac{\sin(kR)}{kR}=\mathcal{V}_b,
\end{equation}
is given by $\mathcal{V}_{b,cr}=\min_{x>0} \dfrac{\sin x}{x} = -m\approx -0.22$.

These relations can be eventually rewritten in terms of conditions on the cell density. In fact, if, for instance,  $b(\rho)=\left(1-\frac{\rho}{\rho_{th}}\right)_+$, then for $\rho_\infty<\rho_{th}$, 
$\mathcal{V}_{b}=-\mathcal{V}\left(\frac{\rho_{th}}{\rho_\infty}-1\right)$ 
and there is instability if $\rho_{\infty}>\rho_{th}\mathcal{V}/(\mathcal{V}+m)$.

The critical wave number is then obtained when 
\begin{equation}\label{kcr_crit}
\left[\dfrac{\sin(kR)}{kR}\right]'_{k=k_{cr}}=0 \quad \Longleftrightarrow \quad \tan (k_{cr}R) =k_{cr}R,
\end{equation}
that is when  $k_{cr}R=\bar{\alpha}\pi$ with $\bar{\alpha}\approx1.43$. Therefore, the critical wave length $\Lambda_{cr}=\dfrac{2\pi}{k_{cr}}$ is such that
\[
\Lambda_{cr}= \dfrac{2}{\bar\alpha}R \approx 1.42 R.
\]
If $\mathcal{V}_b <-m$ the system \eqref{eq:sis.p.lin} is always stable, while if $\mathcal{V}_b \in (-m,0)$ there are unstable  wave numbers.  

\begin{figure}
\centering
\subfigure[$b'(\rho_{\infty})<0$]{\includegraphics[scale=0.48]{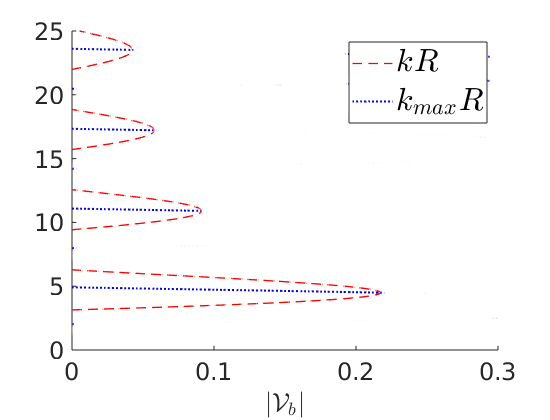}}
\subfigure[$b'(\rho_{\infty})>0$]{\includegraphics[scale=0.48]{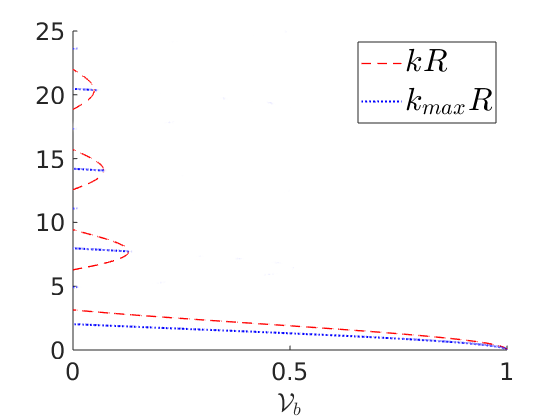}}

\caption{Stability diagram for a localized sensing kernel $\gamma_R=\delta(\lambda-R)$. The red dashed line delimits the unstable region, in (a) when $b'(\rho_{\infty})<0$ and in (b) when $b'(\rho_{\infty})>0$. The blue dotted lines evidentiate the dimensionless wave numbers $k_{max}R$ with local maxima of the growth rate (given by \eqref{kmax_dirac}) in the unstable regime. The lowest curves correspond to the most dangerous wave numbers in both cases.}
\label{fig.dirac}
\end{figure}

Using \eqref{kcr_dirac}, Figure \ref{fig.dirac}a identifies for any $|\mathcal{V}_b|<|\mathcal{V}_{b,cr}|\approx 0.22$ the range of unstable wave numbers and the instability region is the one to the left of the red dashed curve. In this unstable region local maximum growth rates are represented by the blue curves, with the longest waves (corresponding to the lowest curve) being the most unstable ones.

We then have instability of a finite wavelength if for instance, cell speed is sufficiently small, or the turning rate or the sensing radius is sufficiently high, or in the case of volume filling effects, if cell density if sufficiently high

\subsubsection{Case $b'(\rho_{\infty})>0$}
We recall that, from the phenomenological point of view, an increasing $b$  corresponds to a predisposition of cells to re-orient toward regions that are more crowded than in the uniform configuration $\rho_{\infty}$. This might be for instance related to an adhesion-like behaviour for cells that want to stay together. 
An example of this case is $b(\rho)=\rho$,  leading to $\mathcal{B}=1$ and $\mathcal{V}_b=\dfrac{V}{R\mu}$.

In this case the instability condition becomes
\begin{equation}\label{k_dirac_b_pos}
\dfrac{\sin(kR)}{kR}>\mathcal{V}_b>0.
\end{equation}
We may trivially observe that, as $\dfrac{\sin x}{x}  <1$, if $\mathcal{V}_b >1$ the uniform configuration  is always stable. 
On the other hand, if $\mathcal{V}_b \in [0,1)$ wave numbers to the left of the red dashed curves represented in Figure \ref{fig.dirac}(b) are unstable. Again the blue lines represent local maxima for the growth rates with the lowest curve corresponding to the most unstable dimensionless wave number.

\subsection{Uniform sensing kernel}\label{sec:Heaviside}
The same stability analysis can be performed for a sensing function $\gamma_R$ that is a Heaviside function, $\ie$ 
$\gamma_R(\lambda)=H(R-\lambda)$. 
In this case, 
\begin{equation}\label{gamma_Heaviside}
\hat\gamma_R(k)=\dfrac{1-\cos(kR)}{kR}\,.
\end{equation}
Contrary to the localized kernel, for the uniform kernel the first statement of the Proposition holds and the uniform configuration with $\rho=\rho_\infty$ is stable when $\mathcal{B}\le 0$, that is cells at that density prefer to avoid overcrowding.

On the other hand, if $\mathcal{B}>0$ the instability condition \eqref{inst.cond} becomes
\begin{equation}\label{inst.cond.heavi}
\dfrac{1-\cos(kR)}{k^2R^2}>\mathcal{V}_b.
\end{equation}
Therefore, there are unstable waves if $\mathcal{V}_b<\frac{1}{2}$
and, recalling \eqref{maxgrowth}, the maximum growth rate is achieved for $k=k_{max}$ such that 
\begin{equation}\label{kmax_heavi}
\dfrac{\sin(k_{max}R)}{k_{max}R}=2\mathcal{V}_b.
\end{equation}

Referring to Figure \ref{fig.heavi}(a), one then has instability in the region of the $\mathcal{V}_b-kR$ plane to the left of the  red dashed curve with the most unstable wave number identified by the lowest blue curve.

\begin{figure}
\begin{center}
\subfigure[]{\includegraphics[scale=0.45]{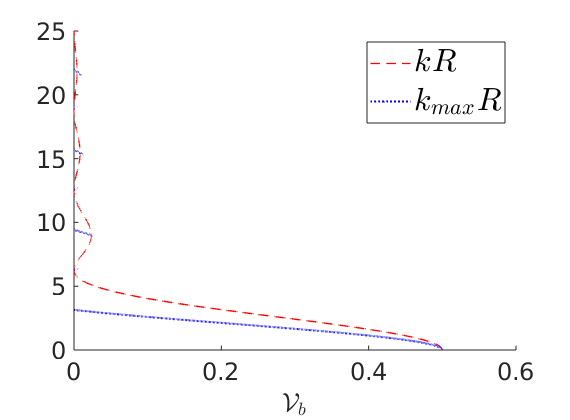}}
\subfigure[]{\includegraphics[scale=0.45]{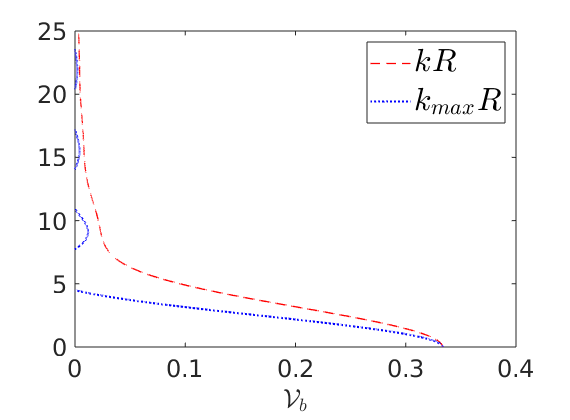}}
\end{center}
\caption{Stability diagram for a uniform (a) and a ramp (b) sensing kernel. The blue dotted line represents $k_{max}R$ given respectively by (a) \eqref{kmax_heavi}  and (b) \eqref{kmax_decr}. The unstable region is the one to the left of the red dashed line, $\ie$ the values of $k_{max}R$ also satisfy \eqref{inst.cond.heavi} in (a) and \eqref{inst.cond.heavi.decr} in (b). }
\label{fig.heavi}
\end{figure}

\subsection{Ramp sensing kernel}\label{sec:ramp}
A sensing function $\gamma_R$ that decreases with the distance from the present position of the cell means that the cell gives more importance to the local information than to distant ones. 
An example is given by the ramp function
\[
\gamma_R(\lambda)=\left(1-\dfrac{\lambda}{R}\right)_+.
\] 
where $(f)_+$ is the positive part of $f$.  In this case, 
\begin{equation}\label{gamma_ramp}
\hat\gamma_R(k)=2\dfrac{kR-\sin(kR)}{k^2R^2},
\end{equation}
Again, independently of the specific form of the decreasing kernel, the uniform configuration with $\rho=\rho_\infty$ is stable when $\mathcal{B}\le 0$, that is, if cells at that density prefer to avoid overcrowding.

On the other hand, if $\mathcal{B}>0$ the instability condition \eqref{inst.cond} becomes
\begin{equation}\label{inst.cond.heavi.decr}
2\dfrac{kR-\sin(kR)}{k^3R^3}>\mathcal{V}_b.
\end{equation}
Therefore, there are unstable waves if $\mathcal{V}_b<\frac{1}{3}$
and, recalling \eqref{maxgrowth}, the maximum growth rate is achieved for $k=k_{max}$ such that 
\begin{equation}\label{kmax_decr}
\dfrac{\sin(k_{max}R)-k_{max}R\cos(k_{max}R)}{k_{max}^3R^3}=\mathcal{V}_b.
\end{equation}

In Figure \ref{fig.heavi}(b), again one then has instability in the region of the $\mathcal{V}_b-kR$ plane to the left of the  red dashed curve with the most unstable wave identified by the lowest blue curve.

\section{Numerical tests}

\begin{figure}[!htbp]
\begin{center}
\subfigure[]{\includegraphics[scale=0.48]{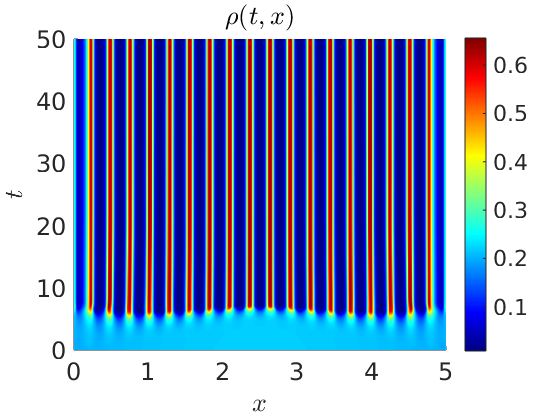}}
\subfigure[]{\includegraphics[scale=0.48]{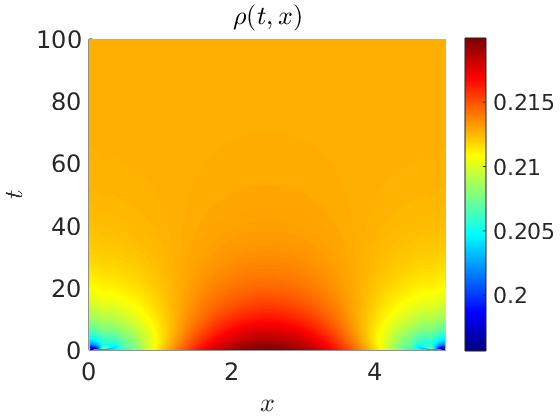}}
\subfigure[]{\includegraphics[scale=0.48]{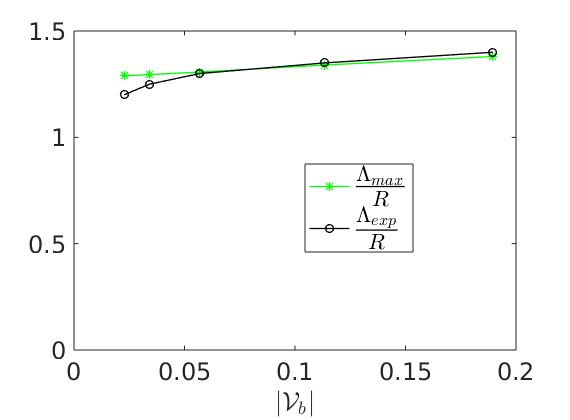}}
\end{center}
\caption{Temporal evolution of $\rho(t,x)$ from $\rho_0(x)=0.2\left(1+0.1\sin(\pi x/5)\right)$ (a) in the unstable case 
$\mathcal{V}_b\approx -0.1125$ and (b) in the stable case $\mathcal{V}_b\approx -0.3376$. 
 (c) Wavelength of the most unstable mode as obtained from the simulation (black line) and from Eq.\eqref{kmax_dirac} (green line).}
\label{simu_1}
\end{figure}

In this section we show some numerical tests in order to show the consistency of our linear stability analysis. We simulate equation \eqref{eq:dir.1D} with specular reflective or periodic boundary conditions, both satisfying Eq. \eqref{noflux}. Specular reflective boudary conditions in one dimension are such that  
\begin{equation}\label{spec_refl}
p^-(t,L,v)=p+(t,L,-v), \qquad p^+(t,0,v)=p^-(t,0,-v).
\end{equation}
We perform a first order splitting for the relaxation and transport step that we perform using a Van Leer scheme. For further details, we address the reader to \citep{Loy_Preziosi} and \citep{Filbet}.
We remark that in our case the $\psi(v|\pm {\bf e})$ is a Gaussian with mean $V^{\pm}$ and variance $s^2$. In particular we consider $s^2=10^{-4}$, so that $\psi$ is close in sense of measure to a Dirac delta and the perturbation of the dispersion relation is of the same order of $s^2$. 

\subsection{Volume filling dynamics}
In order to  mimick the dynamics of cells that are more likely to re-orient where there are less cells and tend to avoid overcrowded areas, corresponding to  the case $b'(\rho)<0$
we set
\[
b(\rho)=\left(1-\dfrac{\rho}{\rho_{th}}\right)_+ .
\]
For well posedness reasons, we shall always consider $\max_{\Omega}\rho_0(x)<\rho_{th}$. Otherwise, physical constraint effects should be taken into account as in \cite{Loy_Preziosi2}. Specifically, we consider a Dirac delta sensing function $\gamma$ and we set $\rho_{th}=0.5$ and an initial condition  
$\rho_0(x)=0.2\left(1+0.1\sin(\pi x/5)\right)$, so that $\rho_{\infty}\approx 0.2127$ corresponding to  
$\mathcal{B}\approx -0.73$. 

Recalling Figure \ref{fig.dirac} we have that when the kernel is a Dirac delta the critical value for $\mathcal{V}_b$ is about $-0.22$.
Having set $V=0.25$ and $R=0.2$, in Figure \ref{simu_1}(a) we use $\mu=15$, so that $\mathcal{V}=1/12$ and everywhere $|\mathcal{V}_b|$ is above the critical value  leading to instability. In fact, its value is $\mathcal{V}_b\approx -0.1125$. 

In Figure \ref{simu_1}(b) $V=0.25$, $R=0.2$ and $\mu=5$, so that $\mathcal{V}=1/4$ and $\mathcal{V}_b\approx -0.3376$. Therefore, the initial density distribution is always below the critical value and the perturbation decays to the homogeneous solution $\rho=\rho_{\infty}\approx 0.2127$. 


Finally, in Figure \ref{simu_1}(c) we compare the theoretical value of the wavelength  of the most unstable mode with the one obtained simulating the system. We find that they are very close, with a discrepancy of about 10\% for small values of $|\mathcal{V}_b|$ closer to zero and a practical coincidence for values above 0.06.

We recall that in this case the first statement of the Proposition does not hold while in the case of smoother non increasing kernels (e.g., the Heaviside or ramp kernels treated below), we always have stability.

\subsection{Adhesion}
We consider now the case in which cells prefer to stay together and reorient toward regions with more crowded areas. Specfically, we take $b(\rho)=\rho$ so that $b'(\rho)=1$ is positive. Therefore, in this case $\mathcal{B}=1$ regardless of the density distribution and then $\mathcal{V}_b=\mathcal{V}$. Having set $V=0.25$ and $R=0.04$, changes in  
$\mathcal{V}_b$ correspond to changes in $\mu$.


\begin{figure}[!htbp]
\begin{center}
\subfigure[localized kernel]{\includegraphics[scale=0.45]{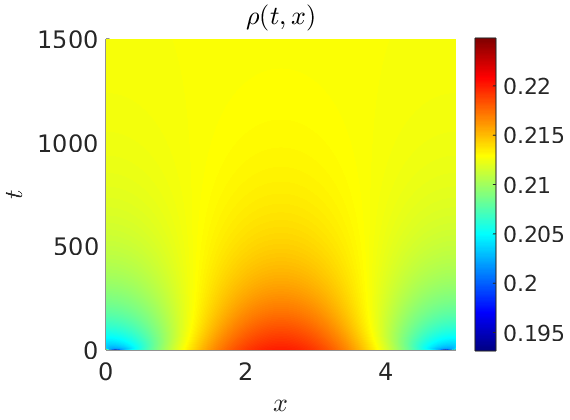}}
\subfigure[localized kernel]{\includegraphics[scale=0.45]{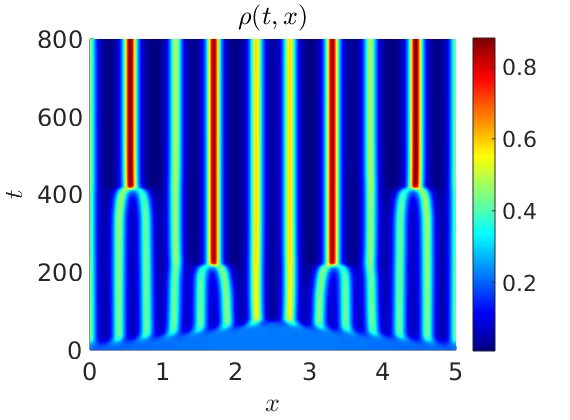}}
\end{center}
\begin{center}
\subfigure[uniform kernel]{\includegraphics[scale=0.45]{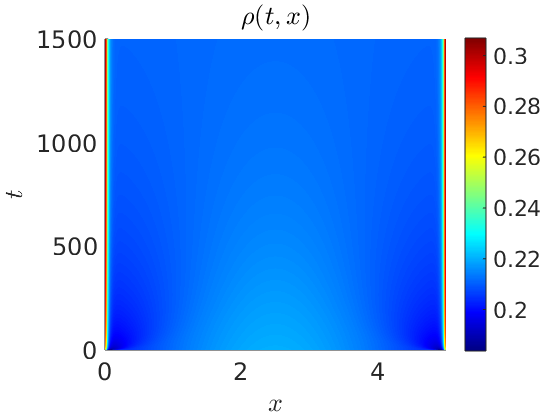}}
\subfigure[uniform kernel]{\includegraphics[scale=0.45]{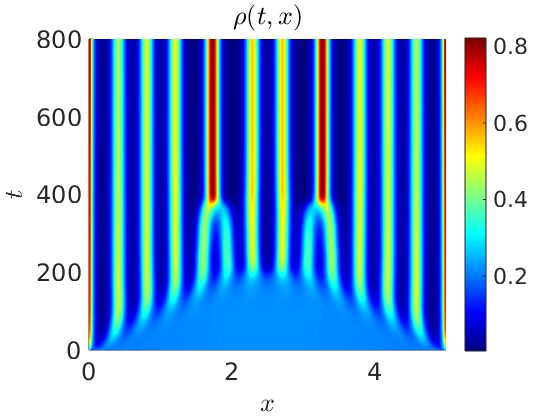}}
\end{center}
\caption{Evolution of the density distribution in the adhesion case starting from the initial condition 
$\rho_0(x)=0.2\left(1+0.1\sin(\pi x/5)\right)$, so that $\rho_{\infty}\approx 0.2127$,
 for a delta kernel (top row) and a Heaviside kernel (bottom row). 
In the left column the values of $\mathcal{V}_b$ correspond to stable cases, while in the right column $\mathcal{V}$ correspond to the unstable case. Specifically, in (a) $\mathcal{V}_b\approx 1.1364$, in (b) $\mathcal{V}_b\approx 0.7812$, in (c) $\mathcal{V}_b\approx 0.5564$, and  in (d) $\mathcal{V}_b\approx 0.3846$.
}
\label{simu_2}
\end{figure}

In Figure \ref{simu_2} we show linear stability and instability in the case of cell-cell adhesion and perfect reflective boundary conditions. The sensing function is a Dirac delta in the top row and a Heaviside function in the bottom row. 

Recalling Figure \ref{fig.dirac}, in Figure \ref{simu_2}(a), $\mu=5.5$, so that $\mathcal{V}_b\approx 1.1364>1$  corresponds to a stable condition, while in Figure \ref{simu_2}(b), $\mu=8$, so that $\mathcal{V}_b\approx  0.7812<1$ leads to an unstable condition.  
Recalling instead Figures \ref{fig.heavi}, in Figure \ref{simu_2}(c), $\mu=180$, so that $\mathcal{V}_b\approx 0.5564>0.5$, corresponds to a stable condition, while in Figure \ref{simu_2}(d), we $\mu=260$, so that $\mathcal{V}_b\approx 0.3846<0.5$ corresponds to an unstable condition. 
 

\begin{figure}[!htbp]
\begin{center}
\subfigure[localized lernel]{\includegraphics[scale=0.2]{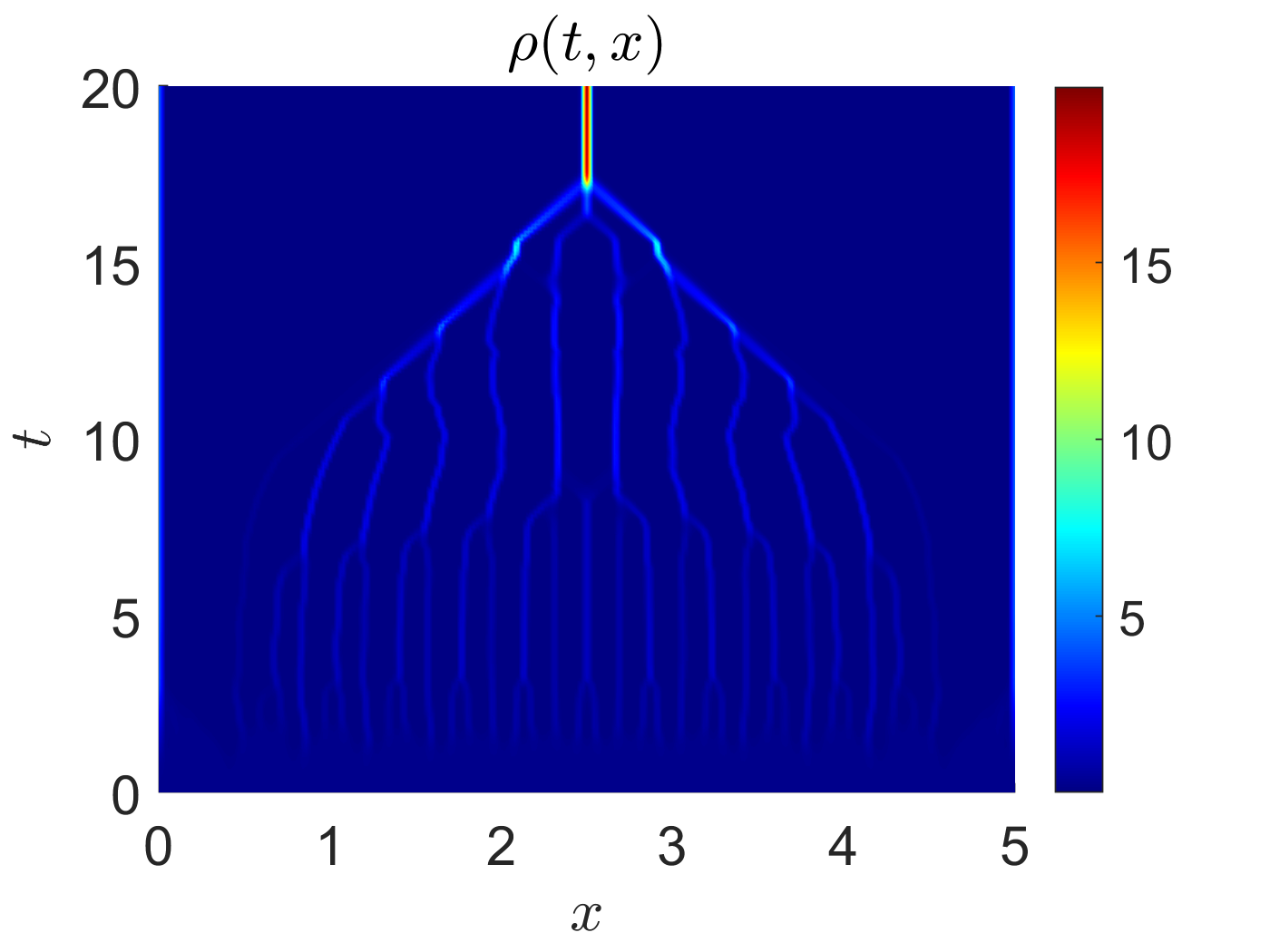}}
\subfigure[uniform kernel]{\includegraphics[scale=0.15]{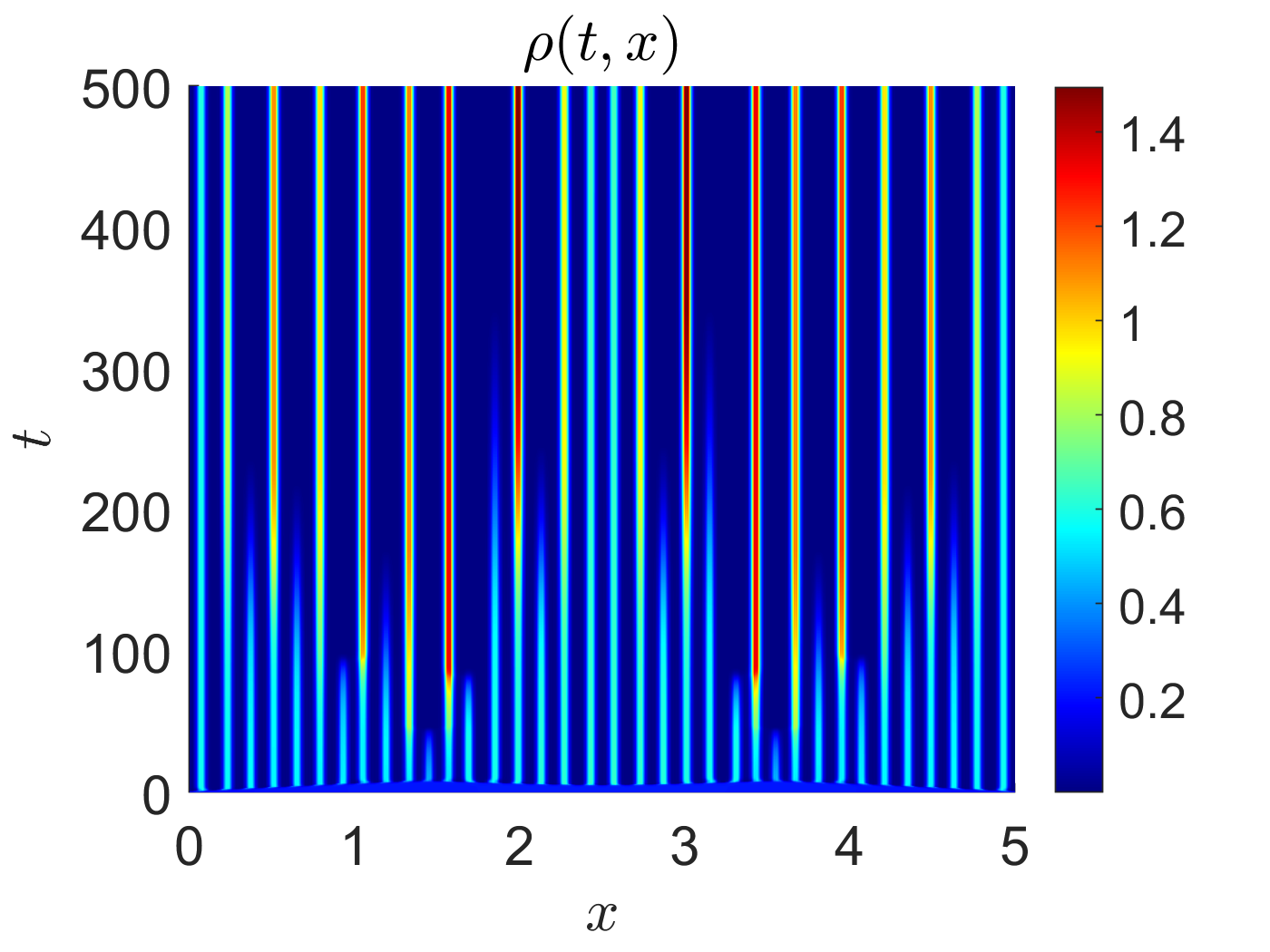}}
\subfigure[decreasing kernel]{\includegraphics[scale=0.5]{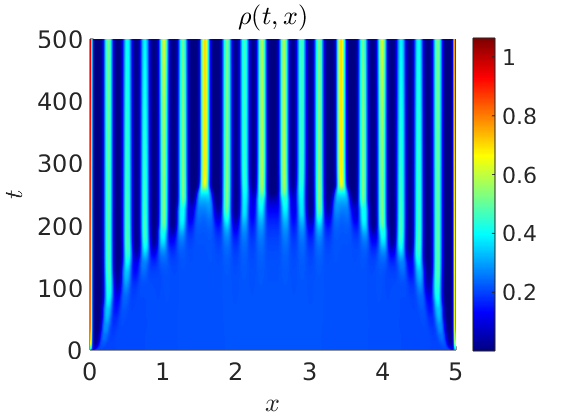}}
\end{center}
\caption{Comparison of unstable evolutions for the same value of  $\mathcal{V}=0.0625$, (given by $V=0.25, R=0.04, \mu=100$) starting form the initial condition $\rho_0(x)=0.2\left(1+0.1\sin(\pi x/5)\right)$, so that $\rho_{\infty}\approx 0.2127$. In (a) $\gamma_R(\lambda)=\delta(\lambda-R)$, in (b) $\gamma_R(\lambda)=H(R-\lambda)$ and in (c) $\gamma_R=\left(1-\frac{\lambda}{R}\right)_+$.}
\label{simu_4}
\end{figure}

Finally, in Figure \ref{simu_4} we compare the different evolution 
in case of three different sensing kernels, namely a Dirac delta, a Heavyside function, and a ramp kernel, for the same value of  $\mathcal{V}=0.0625$ that always fall in the  unstable range.

\subsection{The case of asymmetric $\psi$}

We now consider Eq. \eqref{transport.general} with \eqref{J_r.S} and \eqref{distribution.g.r} allowing the speed probability distribution to depend on $\hv$, i.e.  $\psi=\psi(v|\hv)$. In the one-dimensional case it means that cells that go to the right and to the left have different speed distributions.

Following the same procedure as in Section 3, we obtain again Eq.\eqref{eq:sist.p_gen} and the  local asymptotic equilibria \eqref{asympt}, but with 
with $T^+$ and $T^-$ defined as in 
\eqref{T+} and \eqref{T-} with $\psi(v)$ respectively substituted by $\psi(v|{\bf e})\equiv\psi^+(v)$ and $\psi(v|-{\bf e})\equiv\psi^-(v)$

The local asymptotic equilibria are  stationary and homogeneous if and only if 
\[
p^+_{\infty}(v)=\rho_{\infty}\dfrac{\psi^+(v)}{2}, \qquad p^-_{\infty}(v)=\rho_{\infty}\dfrac{\psi^-(v)}{2},
\]
which, on the contrary of the symmetric case, are no longer equal.

Carrying out the same computations as in Section \ref{Sec:lin.inst}, it is convenient to keep Eq.\eqref{solvability} in the following modified form
\begin{equation}\label{solvability2}
\dfrac{\mu}{2}\int_{\mathbb{R}_+}\left[
\dfrac{1+i\mathcal{B} \hat{\gamma}_R(k)}{(\sigma+\mu)+ikv}\psi^+(v)+
\dfrac{1-i\mathcal{B} \hat{\gamma}_R(k)}{(\sigma+\mu)-ikv}\psi^-(v)\right]
\,dv=1\,,
\end{equation}
being $\hat{\gamma}_R(k)$ defined as in \eqref{gamma_trans}.

In order to understand the stability behaviour we will consider $\psi^\pm(v)=\delta(v-V^\pm)$. 
This allows  to integrate \eqref{solvability2} to get
$$\mu
\left[\sigma+\mu+k\mathcal{B}\hat\gamma_R(k)\frac{V^{+}+V^{-}}{2}+ik\frac{V^+ -V^-}{2}\right]=
(\sigma+\mu)^2+k^2V^+ V^- +ik\frac{V^{+}-V^{-}}{2}(\sigma+\mu)\,,
$$
or
\[
\sigma^2+\sigma \left[ \mu +ik(V^+-V^-)\right]+k^2V^+V^- -\mathcal{B}\mu \hat{\gamma}_R(k) k
\dfrac{ V^+ +V^-}{2}+ i k\mu\dfrac{V^+-V^-}{2} =0\,.
\]
The dispersion relation then reads
\begin{equation}\label{disp.rel.gen.asimm}
\sigma=\dfrac{-\mu-i k \left( V^+-V^-\right)+\sqrt{\mu^2-k^2(V^++V^-)^2+2\mathcal{B}\mu \hat{\gamma}_R(k) k(V^++V^-)}}{2}\,,
\end{equation}
where we notice that the speed difference $(V^+-V^-)$ affects the imaginary part while the real part is affected by the mean speed $(V^++V^-)/2$. Furthermore, the fact that $V^+ \neq V^-$ implies that the wave frequencies $\sigma$ are always complex, hence we expect moving patterns.

According to the type of kernels one then has
\begin{equation}
\begin{array}{lcl}
\gamma_R(\lambda)=\delta(\lambda-R) \quad &\Longrightarrow&\quad
\mathcal{B} \,\dfrac{\sin(kR)}{kR}>\bar{\mathcal{V}}
\,,\\[12pt]
\gamma_R(\lambda)=H(R-\lambda) \quad &\Longrightarrow&\quad
\mathcal{B} \,\dfrac{1-\cos(kR)}{k^2R^2}>\bar{\mathcal{V}}
\,,\\[12pt]
\gamma_R(\lambda)=\left(1-\frac{\lambda}{R}\right)_+ \quad &\Longrightarrow&\quad
\mathcal{B} \,\dfrac{kR-\sin(kR)}{k^3R^3}>\bar{\mathcal{V}}
\,,
\end{array}
\end{equation}
where
\begin{equation}
\bar{\mathcal{V}}=
\dfrac{V^++V^-}{2\mu R}\,.
\end{equation}
that reflect the conditions \eqref{inst.cond.dirac}, \eqref{inst.cond.heavi}, and \eqref{inst.cond.heavi.decr} found in the symmetric case.

\begin{figure}[!htbp]
\begin{center}
\subfigure[Perfect relfection BCs]{\includegraphics[scale=0.3]{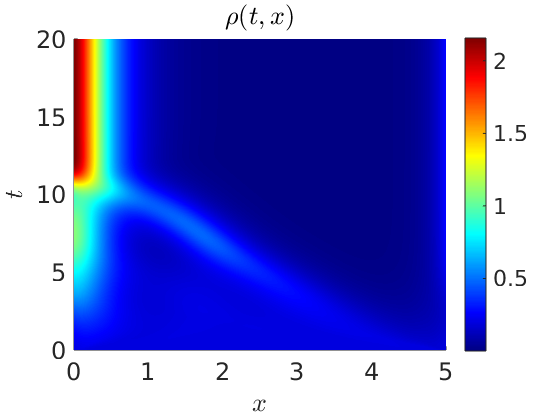}}
\subfigure[Periodic BCs]{\includegraphics[scale=0.3]{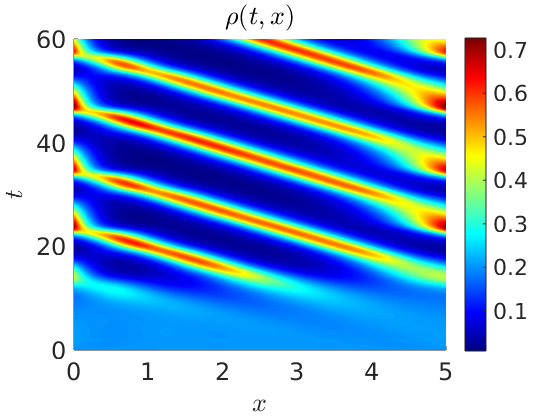}}
\subfigure[]{\includegraphics[scale=0.3]{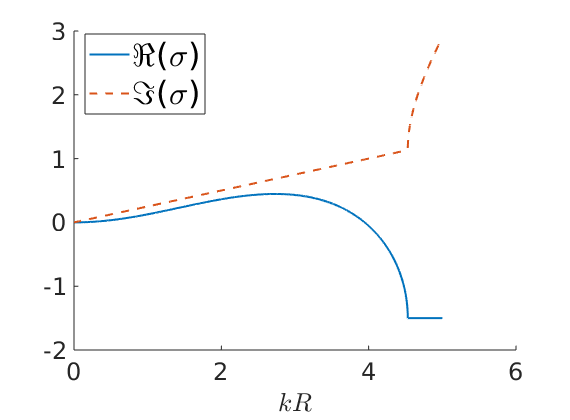}}
\\
\subfigure[Perfect relfection BCs]{\includegraphics[scale=0.3]{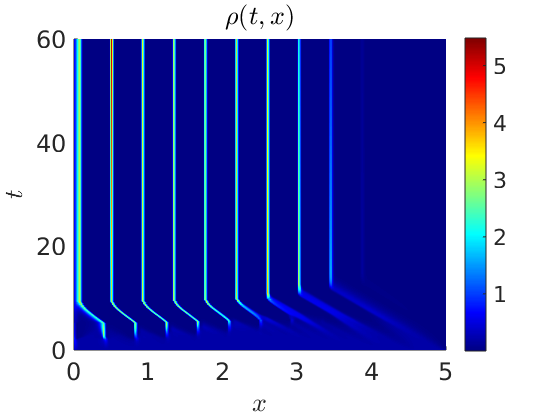}}
\subfigure[Periodic BCs]{\includegraphics[scale=0.3]{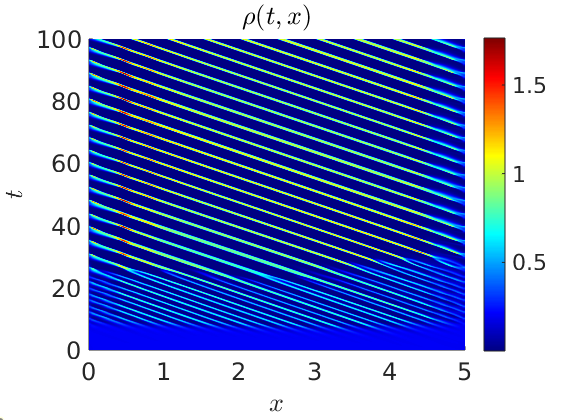}}
\subfigure[]{\includegraphics[scale=0.3]{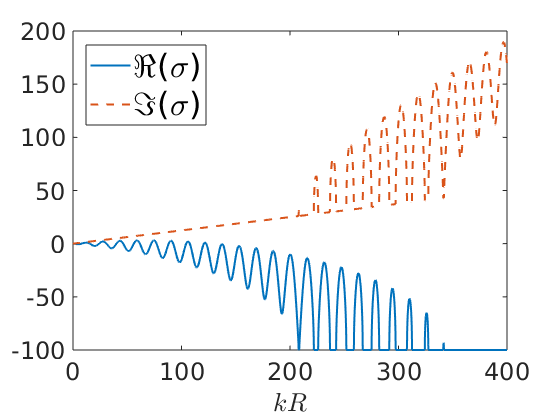}}
\\
\subfigure[Periodic BCs]{\includegraphics[scale=0.3]{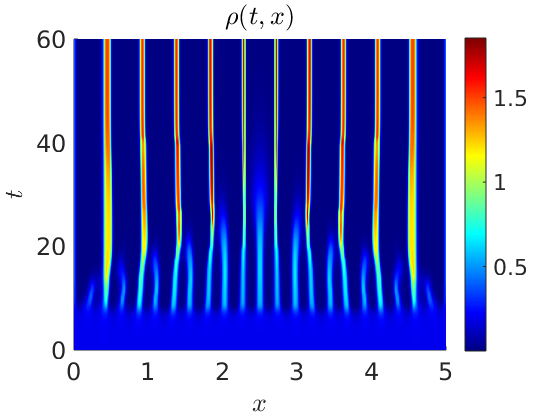}}
\subfigure[Periodic BCs]{\includegraphics[scale=0.3]{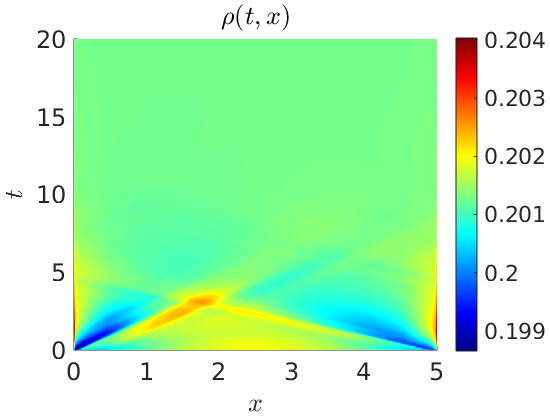}}
\subfigure[]{\includegraphics[scale=0.3]{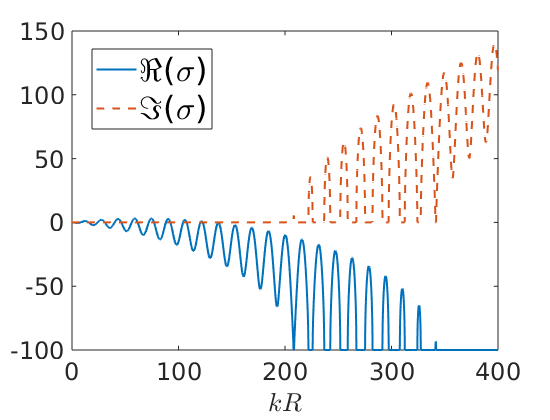}}
\end{center}
\caption{Evolution in the asymmetric case for a localised sensing kernel. The initial condition is always $\rho_0(x)=0.2\left(1+0.01\sin(\pi  x/5)\right)$, so that $\rho_{\infty}\approx 0.2013$. 
(a)-(c): Adhesion, $\mu=3, R=0.44, V^+=1, V^-=0.5, \mathcal{V}\approx _b 0.5682$. 
(d)-(f):  Volume filling, $\mu=200, R=0.4, V^+=0.25,V^-=0.5,\rho_{th}=1, \mathcal{V}_b\approx -0.018$. 
(g):   Volume filling, $\mu=200, R=0.4, V^+=V^-=0.375,\rho_{th}=1,  \mathcal{V}_b\approx -0.018$.
(h)-(i): Adhesion, $\mu=1, R=0.44, V^+=0.5, V^-=1, \mathcal{V}_b\approx  1.7045$.
 }
\label{simu_5}
\end{figure}

In Figure \ref{simu_5} we present some tests in the asymmetric case both with perfectly reflective boundary conditions and periodic boundary conditions in the case of a Dirac delta sensing function. In (a)-(c) we consider adhesion, i.e. $b(\rho)=\rho$ in an unstable situation. In particular, Figure \ref{simu_5}(c) shows that the imaginary part of the eigenvalue is always stricly positive and in fact both in (a) and (b) there is a moving pattern. In (d)-(f) we have unstable configurations in the case of volume filling Figure \ref{simu_5}(f) shows again that as $V^+ \neq V^-$ there are complex eigenvalues and therefore we have a moving pattern in the direction $-{\bf  e}$ that is what we expect as $V^+<V^-$. In Figure \ref{simu_5}(g) we present the test with the same value of $\mathcal{V}_b$ as in Figure \ref{simu_5}(d)-(e) but with $V^+=V^-$.  As $V^+=V^-$ in Figure \ref{simu_5}(g), the pattern is symmetric and there are real eigenvalues (see Figure (i)). In Figure \ref{simu_5}(g) the boundary conditions are periodic, but, as $V^+=V^-$, we would obtain the same result with specular reflection boundary conditions. In Figure \ref{simu_5}(h) we instead consider adhesion and we have a stable configuration as $\mathcal{V}_b>1$, even if we have an anisotropic setting as $V^+ \neq V^-$: the transient shows an asymmetric behavior, but the solution goes to a stationary homogeneous case.

\section{Discussion}
We analyzed the stability properties of a non-local kinetic equation implementing a velocity jump process in which the transition probability models a directional response to a non-local evaluation of the macroscopic cell density. We identified a stability condition \eqref{inst.cond} that depends on two non-dimensional parametres: $\mathcal{B}$ that takes into account of the directional response  through a non-local measure of the cell density weighted by a sensing kernel $b(\rho)$, and $\mathcal{V}$, that takes into account of the motility properties of cells, specifically their mean speed, sensing radius and tumbling frequency. We proved that if $\mathcal{B}$ is negative, corresponding to a response that tends to avoid crowding, then  instability can occur only for sensing functions that are not non increasing and for $\mathcal{V}/|\mathcal{B}|$ small enough. On the other hand, if $\mathcal{B}$ is positive, corresponding to an adhesion-like behavior, then the homogeneous solution is unstable to long wave for values of $\mathcal{V}/\mathcal{B}$ that are sufficiently small.
Considering that $\mathcal{V}=\frac{V}{\mu R}$, this means that for fixed $\mathcal{B}$ instability occurs in stiff regimes, e.g. a large tumbling frequency leads to instability. 

On the other hand, if $\mathcal{V}$ is fixed, the stiffness of the response $b(\rho_{\infty})$ determines the transition from stability to instability. In the case of increasing $b$ at $\rho_{\infty}$, cells tend to go towards zones that are more crowded. If $b'(\rho_{\infty}$) is sufficiently high, there is linear instability, for all the analyzed sensing functions. If $b'(\rho_{\infty})$ is negative, cells tend to go towards regions that are less crowded. In this case instability can be triggered  for example in the case in which the sensing function is a Dirac delta, $\ie$ when cells do not evaluate the density in between their position and the sensing position $x\pm R$.

Numerical simulations show that the linear stability analysis predicts pattern formation (or stability) quite sharply and it is able to catch the characteristic wavelengths of the pattern.

We restricted to the case in which cell density only affects the direction of motion, e.g., cells turn away if they sense an overcrowded area. Of course, cell density can also affect their speed. This case is currently under study. Actually, the most interesting case is the one in which cell density has an ambivalent effect. For instance, cells are attracted by the other cells because of cell-cell adhesion but at the same time they want to stay away from overcrowding and they do not want to overcome a certain threshold density (volume filling).

\section*{Acknowledgements}
This work was partially supported by Istituto Nazionale di Alta Matematica, Ministry of Education, Universities and Research, through the MIUR grant Dipartimenti di Eccellenza 2018-2022, Project no. E11G18000350001, and the Scientific Reseach Programmes of Relevant National Interest project n. 2017KL4EF3. NL also acknowledges Compagnia di San Paolo that funds her Ph.D. scholarship.

\bibliography{references}
\newpage

\end{document}